\documentclass[a4paper,UKenglish,cleveref, autoref, thm-restate]{lipics-v2021}

\usepackage{hyperref} 
\pdfoutput=1 

\title{Formalized Haar Measure}

\author{Floris van Doorn}{University of Pittsburgh, USA}{fpvdoorn@gmail.com}{https://orcid.org/0000-0003-2899-8565}{The author is supported by the Sloan Foundation (grant G-2018-10067)}

\authorrunning{F. van Doorn}

\Copyright{Floris van Doorn}



\category{} 

\relatedversion{} 

\acknowledgements{I want to thank Jeremy Avigad and Tom Hales for proofreading this paper.}

\nolinenumbers 

\hideLIPIcs  

\EventEditors{John Q. Open and Joan R. Access}
\EventNoEds{2}
\EventLongTitle{42nd Conference on Very Important Topics (CVIT 2016)}
\EventShortTitle{CVIT 2016}
\EventAcronym{CVIT}
\EventYear{2016}
\EventDate{December 24--27, 2016}
\EventLocation{Little Whinging, United Kingdom}
\EventLogo{}
\SeriesVolume{42}
\ArticleNo{23}

\usepackage{fontawesome} 
\usepackage{mathtools} 
\usepackage{upgreek} 
\usepackage{MnSymbol} 

\newcommand{\R}{\ensuremath{\mathbb{R}}}
\newcommand{\N}{\ensuremath{\mathbb{N}}}
\newcommand{\C}{\ensuremath{\mathbb{C}}}
\renewcommand{\epsilon}{\varepsilon}

\usepackage{xcolor}
\definecolor{keywordcolor}{rgb}{0.7, 0.1, 0.1}   
\definecolor{commentcolor}{rgb}{0.4, 0.4, 0.4}   
\definecolor{symbolcolor}{rgb}{0.0, 0.1, 0.6}    
\definecolor{sortcolor}{rgb}{0.1, 0.5, 0.1}      
\definecolor{errorcolor}{rgb}{1, 0, 0}           
\definecolor{stringcolor}{rgb}{0.5, 0.3, 0.2}    

\lstMakeShortInline"
\newcommand{\link}{\ensuremath{{}^\text{\faExternalLink}}}
\DeclareMathOperator{\Hom}{Hom}

\begin{document}

\maketitle

\begin{abstract}
We describe the formalization of the existence and uniqueness of Haar measure
in the Lean theorem prover.
The Haar measure is an invariant regular measure on
locally compact groups, and it has not been formalized in a proof assistant before.
We will also discuss the measure theory library in Lean's mathematical library \textsf{mathlib},
and discuss the construction of product measures
and the proof of Fubini's theorem for the Bochner integral.
\end{abstract}

\section{Introduction}
\label{section:introduction}
Measure theory is an important part of mathematics,
providing a rigorous basis for integration and probability theory.
The main object of study in this part is the concept of a measure,
that assigns a size to the measurable subsets of the space in question.
The most widely used measure is the Lebesgue measure $\lambda$ on $\R$ (or $\R^n$)
with the property that $\lambda([a,b]) = b - a$. 
The Lebesgue measure is \emph{translation invariant},
meaning that translating a measurable set does not change its measure.

An important generalization of the Lebesgue measure is the Haar measure,
which provides a measure on locally compact groups.
This measure is also invariant under translations, i.e.~applying the group operation.
For non-abelian groups we have to distinguish between left and right Haar measures,
which are invariant when applying the group operation on the left and right, respectively.
A Haar measure is essentially unique,
which means that any invariant measure is a constant multiple of a chosen Haar measure.
The Haar measure is vital in various areas of mathematics,
including harmonic analysis, representation theory,
probability theory and ergodic theory.

In this paper we describe the formal definition of the Haar measure
in the Lean theorem prover~\cite{moura2015lean}, and prove its uniqueness.
We build on top of the Lean mathematical library \textsf{mathlib}~\cite{mathlib}.
We heavily use the library for measure theory in \textsf{mathlib}, improving it in the process.
As part of this formalization we expanded the library with products measures and Fubini's theorem,
along with many other contributions to existing libraries.
The results in this paper have been fully incorporated into \textsf{mathlib}.

Many proof assistants have a library of measure theory and integration,
including Isabelle~\cite{richter2004integration,holzl2011measure},
HOL/HOL4~\cite{hurd2003formal,coble2010anonymity,mhamdi2013lebesgue},
HOL Light~\cite{harrison2013euclidean},
Mizar~\cite{endou2008lebesgue} 
and PVS~\cite{lester2007topology}.
Most of these libraries focus on the Lebesgue measure and the Lebesgue integral,
and on applications to probability theory or analysis in Euclidean spaces.

The Isabelle/HOL library has an especially good development of measure theory,
including product measures and the Bochner integral~\cite{avigad2017clt}.
The Lean library also has these concepts,
and the development of them was heavily inspired by the Isabelle/HOL library.
However, to my knowledge, the Haar measure has not been formalized in any proof assistant other than Lean.

In this paper we will discuss the measure theory library in Lean
(in Section~\ref{section:measure-theory}),
which was already in place before this formalization started,
and was developed by Johannes Hölzl, Mario Carneiro, Zhouhang Zhou and Yury Kudryashov, among others.
The other sections describe contributions that are new as of this formalization:
product measures (in Section~\ref{section:products}),
the definition of the Haar measure (in Section~\ref{section:haar-def}),
and the uniqueness of the Haar measure (in Section~\ref{section:uniqueness}).

We will link to specific results in \textsf{mathlib} using the icon
\href{https://github.com/leanprover-community/mathlib/blob/2cbaa9cc29bc44812b679810352ed383f35dcf75/src/measure_theory/haar_measure.lean#L537}{\link}.
To ensure that these links will work while \textsf{mathlib} is updated,
we link to the version of \textsf{mathlib} as of writing this paper.
Readers of this paper are also encouraged to use the search function
of the \textsf{mathlib} documentation pages
\href{https://leanprover-community.github.io/mathlib_docs/}{\link}
but we do not link to specific pages of the documentation,
as these links are prone to break as the library evolves.

We used different sources for this formalization.
We started out using the notes by Jonathan Gleason~\cite{gleason2010existence}
that gave a construction of the Haar measure in detail.
However, one of the main claims in that paper was incorrect,
as described in Section~\ref{section:haar-def}.
We then used the books by Halmos and Cohn~\cite{halmos1950measure,cohn2013measure} to fill in various details.
For product measures we followed Cohn~\cite{cohn2013measure}.
However, we could not find a good source for the proof of Fubini's theorem for the Bochner integral,
other than the formalization in Isabelle.
For this reason, we will provide detailed proofs of that result in Section~\ref{section:products}.
For the uniqueness of Haar measure we followed Halmos~\cite{halmos1950measure}.
While we define the Haar measure for any locally compact group,
we prove the uniqueness only for second-countable locally compact groups.
The reason is that Halmos gives a uniqueness proof using measurable groups,
and only second-countable topological groups form measurable groups.
In the proof of a key theorem of Halmos there was a gap,
which we fixed in the formalization by changing the statement of that theorem.

\section{Preliminaries}
\label{section:preliminaries}
\subsection{Lean and \textsf{mathlib}}
\label{section:lean}

Lean~\cite{moura2015lean} is a dependently typed proof assistant
with a strong metaprogramming framework~\cite{ebner2017metaprogramming}.
It has a large community-driven mathematical library, called \textsf{mathlib}.
We use results from various areas in this library,
including topology, analysis and algebra.
In particular the existing library for topological groups was very helpful for this formalization.
The fact that Lean provides a very expressive type theory was convenient
in the definition of the Bochner integral,
as we will discuss in Section~\ref{section:measure-theory}.

\textsf{mathlib} uses \emph{type-classes}~\cite{wadler1989polymorphism,spitters2011typeclasses}
to organize the mathematical structures associated to various types,
using the partially bundled approach. For example, when "G" is a type,
the statement \lstinline{[group G]} states that "G" has the structure of a group,
and \lstinline{[topological_space G]} states that "G" has the structure of a topological space.
The square brackets mean that the argument will be implicitly inserted in definitions and theorems
using type-class inference.
Some type-classes are \emph{mixins}, and take other type-classes as arguments.
For example, the class \lstinline{topological_group} takes as arguments classes for
\lstinline{group} and \lstinline{topological_space} (which are implicit).
This means that we need to write
\lstinline{[group G] [topological_space G] [topological_group G]}
to state that "G" is a topological group.
Similarly, predicates on topological spaces ---
such as \lstinline{t2_space}, \lstinline{locally_compact_space} or \lstinline{compact_space} ---
take the topological space structure as argument.
This design decision makes it somewhat verbose to get the right structure on a type, but has the
advantage that these components can be freely mixed as needed.

Lean distinguishes groups written multiplicatively ("group")
from groups written additively ("add_group").
This is convenient since it automatically establishes multiplicative notation ("1 * x⁻¹")
or additive notation ("0 + -x" or "0 - x") when declaring a group structure.
Also, most theorems exist for both versions; the property "x * y = y * x" is called "mul_comm"
and the property "x + y = y + x" is "add_comm".
To avoid code duplication, there is an attribute "@[to_additive]" that,
when applied to a definition or theorem for multiplicative groups,
automatically generates the corresponding declaration for additive groups
by replacing all multiplicative notions with the corresponding additive ones.

Some code snippets in this paper have been edited slightly for the sake of readability.

\subsection{Mathematics in \textsf{mathlib}}
\label{section:math-preliminaries}

The notation for the image of a set \lstinline{A : set X} under a function \lstinline{f : X → Y}
is \lstinline/f '' A = {y : Y | ∃ x ∈ A, f x = y}/ and the preimage of \lstinline{B : set Y} is
\lstinline/f ⁻¹' B = {x : X | f x ∈ B}/. The notation \lstinline{λ x, f x} is used for the function $x \mapsto f(x)$. We write "(x : X)" for the element "x" with its type "X" given explicitly.
This can also be used to \emph{coerce} "x" to the type "X".

\textsf{mathlib} has a large topology library,
where many concepts are defined in terms of filters.
The following code snippet contains some notions in topology used in this paper.
\begin{lstlisting}
class t2_space (X : Type*) [topological_space X] : Prop :=
(t2 : ∀ x y, x ≠ y → ∃ U V : set X, is_open U ∧ is_open V ∧
  x ∈ U ∧ y ∈ V ∧ U ∩ V = ∅)

def is_compact {X : Type*} [topological_space X] (K : set X) :=
∀ ⦃f : filter X⦄ [ne_bot f], f ≤ 𝓟 K → ∃ x ∈ K, cluster_pt x f

class locally_compact_space (X : Type*) [topological_space X] : Prop :=
(local_compact_nhds : ∀ (x : X) (U ∈ 𝓝 x), ∃ K ∈ 𝓝 x, K ⊆ U ∧
  is_compact K)

class second_countable_topology : Prop :=
(is_open_generated_countable : ∃ C : set (set X), countable C ∧
  t = generate_from C)

class complete_space (X : Type*) [uniform_space X] : Prop :=
(complete : ∀ {f : filter α}, cauchy f → ∃ x, f ≤ 𝓝 x)

class topological_group (G : Type*) [topological_space G] [group G] :=
(continuous_mul : continuous (λ p : G × G, p.1 * p.2))
(continuous_inv : continuous (inv : G → G))
\end{lstlisting}
Some of these definition of a compact set might look unfamiliar, formulated in terms of filters,
using in particular the principal filter "𝓟" and the neighborhood filter "𝓝".
For example, the second definition states that a set "K" is \emph{compact}~\href{https://github.com/leanprover-community/mathlib/blob/2cbaa9cc29bc44812b679810352ed383f35dcf75/src/topology/subset_properties.lean#L64}{\link}\
if any filter "f ≠ ⊥" containing "K" has a cluster point "x" in "K",
which means that "𝓝 x ⊓ f ≠ ⊥".
It is equivalent to the usual notion of compactness
that every open cover of "K" has a finite subcover.~\href{https://github.com/leanprover-community/mathlib/blob/2cbaa9cc29bc44812b679810352ed383f35dcf75/src/topology/subset_properties.lean#L300}{\link}\
A complete space is a space where every Cauchy sequence converges,~\href{https://github.com/leanprover-community/mathlib/blob/2cbaa9cc29bc44812b679810352ed383f35dcf75/src/topology/uniform_space/cauchy.lean#L209}{\link}\
which is formulated in the general setting of a \emph{uniform spaces},
which simultaneously generalizes metric spaces and topological groups.
We define a \emph{locally compact group} to be a topological group in which the topology is both
locally compact and Hausdorff ($T_2$)

Lean's notion of infinite sums is defined for absolutely convergent series
as the limit of all finite partial sums:~\href{https://github.com/leanprover-community/mathlib/blob/2cbaa9cc29bc44812b679810352ed383f35dcf75/src/topology/algebra/infinite_sum.lean#L57}{\link}\
\begin{lstlisting}
def has_sum [add_comm_monoid M] [topological_space M]
  (f : I → M) (x : M) : Prop :=
tendsto (λ s : finset I, ∑ i in s, f i) at_top (𝓝 x)
\end{lstlisting}
If "M" is Hausdorff, then the sum is unique if it exists,
and we denote it by "∑' i, f i".
If the series does not have a sum, it is defined to be "0".

The type "ennreal"~\href{https://github.com/leanprover-community/mathlib/blob/2cbaa9cc29bc44812b679810352ed383f35dcf75/src/data/real/ennreal.lean#L80}{\link}\ is the type of nonnegative real numbers extended by a new element $\infty$,
which we will also denote $[0,\infty]$ in this paper.

A \emph{normed space}~\href{https://github.com/leanprover-community/mathlib/blob/2cbaa9cc29bc44812b679810352ed383f35dcf75/src/analysis/normed_space/basic.lean#L911}{\link}\ $X$ is a vector space over a normed field $F$
equipped with a norm $\|{\cdot}\|:X\to\R$ satisfying
\begin{itemize}
\item $(X, d)$ is a metric space, with the distance function given by $d(x,y) \vcentcolon= \|x - y\|$.
\item for $r \in F$ and $x \in X$ we have $\|r\cdot x\| = \|r\|\cdot \|x|$.
\end{itemize}
In this paper, we will only consider normed spaces over $\R$,
which is a normed field with its norm given by the absolute value.
A \emph{Banach space} is a complete normed space.
In Lean we do not explicitly define Banach spaces, but instead talk about complete normed spaces.

\section{Measure Theory in \textsf{mathlib}}
\label{section:measure-theory}

In this section we describe the background measure theory library used in this formalization.

A basic notion in measure theory is a \emph{$\upsigma$-algebra} on $X$,
which is a nonempty collection of subsets of $X$
that is closed under complements and countable unions.
In \textsf{mathlib} this is formulated as a type-class, with name \lstinline{measurable_space}:~\href{https://github.com/leanprover-community/mathlib/blob/2cbaa9cc29bc44812b679810352ed383f35dcf75/src/measure_theory/measurable_space.lean#L78}{\link}\
\begin{lstlisting}
class measurable_space (X : Type*) :=
(is_measurable : set X → Prop)
(is_measurable_empty : is_measurable ∅)
(is_measurable_compl : ∀ A, is_measurable A → is_measurable Aᶜ)
(is_measurable_Union : ∀ A : ℕ → set X, (∀ i, is_measurable (A i)) →
  is_measurable (⋃ i, A i))
\end{lstlisting}
We say that a function is \emph{measurable}
if the preimage of every measurable set is measurable:~\href{https://github.com/leanprover-community/mathlib/blob/2cbaa9cc29bc44812b679810352ed383f35dcf75/src/measure_theory/measurable_space.lean#L478}{\link}\
\begin{lstlisting}
def measurable [measurable_space X] [measurable_space Y] (f : X → Y) :=
∀ ⦃B : set Y⦄, is_measurable B → is_measurable (f ⁻¹' B)
\end{lstlisting}
Note that we use "is_measurable" for sets and "measurable" for functions.

In \textsf{mathlib}, measures are defined as special case of outer measures.
An outer measure on $X$ is a monotone function $m:\mathcal{P}(X)\to[0,\infty]$ such that
$m(\emptyset)=0$ and $m$ is countably subadditive:~\href{https://github.com/leanprover-community/mathlib/blob/2cbaa9cc29bc44812b679810352ed383f35dcf75/src/measure_theory/outer_measure.lean#L60}{\link}\
\begin{lstlisting}
structure outer_measure (X : Type*) :=
(measure_of : set X → ennreal)
(empty : measure_of ∅ = 0)
(mono : ∀ {A₁ A₂}, A₁ ⊆ A₂ → measure_of A₁ ≤ measure_of A₂)
(Union_nat : ∀ (A : ℕ → set X),
  measure_of (⋃ i, A i) ≤ (∑' i, measure_of (A i)))
\end{lstlisting}
A measure on a measurable space is an outer measure with two extra properties:
it is countably additive on measurable sets,
and given the value on measurable sets,
the outer measure is the maximal one that is compatible with these values
(which is called \lstinline{trim} in the snippet below).~\href{https://github.com/leanprover-community/mathlib/blob/2cbaa9cc29bc44812b679810352ed383f35dcf75/src/measure_theory/measure_space.lean#L97}{\link}\
\begin{lstlisting}
structure measure (X : Type*) [measurable_space X]
  extends outer_measure X :=
(m_Union : ∀ ⦃A : ℕ → set X⦄,
  (∀ i, is_measurable (A i)) → pairwise (disjoint on A) →
  measure_of (⋃ i, A i) = ∑' i, measure_of (A i))
(trimmed : to_outer_measure.trim = to_outer_measure)
\end{lstlisting}
This definition has two very convenient properties:
\begin{enumerate}
  \item We can apply measures to any set,
  without having to provide a proof that the set is measurable.
  \item Two measures are equal when they are equal on measurable sets.~\href{https://github.com/leanprover-community/mathlib/blob/2cbaa9cc29bc44812b679810352ed383f35dcf75/src/measure_theory/measure_space.lean#L149}{\link}\
\end{enumerate}
Given a measure $\mu$ on $X$ and a map $f:X\to Y$,
we can define the \emph{pushforward} $f_*\mu$,~\href{https://github.com/leanprover-community/mathlib/blob/2cbaa9cc29bc44812b679810352ed383f35dcf75/src/measure_theory/measure_space.lean#L710}{\link}\ which is a measure on $Y$,
defined on measurable sets $A$ by $(f_*\mu)(A)=\mu(f^{-1}(A))$.~\href{https://github.com/leanprover-community/mathlib/blob/2cbaa9cc29bc44812b679810352ed383f35dcf75/src/measure_theory/measure_space.lean#L718}{\link}\

We proceed to define lower Lebesgue integration $\int^-$ for functions $f:X\to[0,\infty]$,
where $X$ is a measurable space. We annotate this integral with a minus sign to distinguish it
from the Bochner integral.
This is done first for \emph{simple functions},
i.e.~functions with finite range
with the property that the preimage of all singletons are measurable.
The integral of a simple function $g$ is simply~\href{https://github.com/leanprover-community/mathlib/blob/2cbaa9cc29bc44812b679810352ed383f35dcf75/src/measure_theory/integration.lean#L541}{\link}\
\[\int^- g\, \mathrm{d}\mu= \int^- g(x)\, \mathrm{d}\mu(x)=
\sum_{y \in g(X)} \mu(g^{-1}\{y\})\cdot y \in [0,\infty].\]
If $f:X\to[0,\infty]$ is any function,
we can define the \emph{(lower) Lebesgue integral} of $f$
to be the supremum of $\int^- g\, \mathrm{d}\mu(x)$
for all simple functions $g$ that satisfy $g\le f$ (pointwise).~\href{https://github.com/leanprover-community/mathlib/blob/2cbaa9cc29bc44812b679810352ed383f35dcf75/src/measure_theory/integration.lean#L850}{\link}\
In \textsf{mathlib} we denote the Lebesgue integral of "f" by \lstinline{∫⁻ x, f x ∂μ}.
We prove the standard properties of the Lebesgue integral for nonnegative functions,
such as the monotone convergence theorem:~\href{https://github.com/leanprover-community/mathlib/blob/2cbaa9cc29bc44812b679810352ed383f35dcf75/src/measure_theory/integration.lean#L997}{\link}\
\begin{lstlisting}
theorem lintegral_supr {f : ℕ → X → ennreal}
  (hf : ∀ n, measurable (f n)) (h_mono : monotone f) :
  (∫⁻ x, ⨆ n, f n x ∂μ) = (⨆ n, ∫⁻ x, f n x ∂μ)
\end{lstlisting}

After defining the Lebesgue integral for nonnegative functions,
most books and formal libraries continue to define
the Lebesgue integral for (signed) functions $f:X\to\R$.
This is defined as the difference of the integrals of the positive and negative parts of $f$.
However, this is not very general: separate definitions are needed
for integrals for functions with codomain $\C$ or $\R^n$.
Instead, we opt for the more general definition of the \emph{Bochner integral},
which is defined for function that map into any second-countable real Banach space.

If $X$ is a topological space, we equip it with the \emph{Borel} $\upsigma$-algebra, which is the
smallest $\upsigma$-algebra that contains all open sets (or equivalently, all closed sets).~\href{https://github.com/leanprover-community/mathlib/blob/2cbaa9cc29bc44812b679810352ed383f35dcf75/src/measure_theory/borel_space.lean#L151}{\link}\
The fact that $X$ is equipped with the Borel $\upsigma$-algebra is a mixin,
and is specified by providing the arguments
\lstinline{[topological_space X] [measurable_space X] [borel_space X]}.

Let $X$ be a measurable space and $E$ be a second-countable real Banach space.
We define two quotients on functions from $X$ to $E$.

Let $X \to_{\mu} E$ be the measurable maps $X \to E$
modulo $\mu$-a.e.~equivalence, i.e.~modulo the relation
\[f =_\mu g \Longleftrightarrow \mu(\{x \mid f(x)\ne g(x)\})=0.~\href{https://github.com/leanprover-community/mathlib/blob/2cbaa9cc29bc44812b679810352ed383f35dcf75/src/measure_theory/ae_eq_fun.lean#L93}{\link}\ \]
Note that we are using the fact that Lean is based on dependent type theory here, since
$X \to_{\mu} E$ is a type that depends on $\mu$.
Here our formalization differs from the one in Isabelle,
since in Isabelle this type cannot be formed.
Instead, a similar argument is given purely on functions,
without passing to the quotient.

If $f : X \to E$ we say that $f$ is \emph{integrable} if~\href{https://github.com/leanprover-community/mathlib/blob/2cbaa9cc29bc44812b679810352ed383f35dcf75/src/measure_theory/l1_space.lean#L383}{\link}\
\begin{itemize}
  \item $f$ is $\mu$-a.e.~measurable,
  i.e.~there is a measurable function $g : X \to E$ with $f=_\mu g$;%
  \footnote{Until recently, the definition of ``integrable'' in \textsf{mathlib} required $f$ to be measurable.
  The definition was modified to increase generality.
  The motivating reason was that we would regularly take integrals of a function $f$
  with a measure $\mu|_A$ defined by $\mu|_A(B)=\mu(A\cap B)$.
  A function that is (for example) continuous on $A$ is also $\mu|_A$-a.e.~measurable,
  and we usually do not need any information about the behavior of $f$ outside $A$, so
  it is not necessary to assume that $f$ is measurable everywhere.}
  and
  \item $\int^- \|f(x)\|\, \mathrm{d}\mu(x)<\infty.$
\end{itemize}
The $L^1$-space $L^1(X,\mu;E)$ is defined
as those equivalence classes of functions in $X \to_{\mu} E$ that are integrable.~\href{https://github.com/leanprover-community/mathlib/blob/2cbaa9cc29bc44812b679810352ed383f35dcf75/src/measure_theory/l1_space.lean#L633}{\link}\
$L^1(X,\mu;E)$ is a normed space,~\href{https://github.com/leanprover-community/mathlib/blob/2cbaa9cc29bc44812b679810352ed383f35dcf75/src/measure_theory/l1_space.lean#L713}{\link}\ with the norm given by
\[\|f\|\vcentcolon=\int^- \|f(x)\|\, \mathrm{d}\mu(x).\]

We define the Bochner integral first for simple functions,
similar to the definition for the Lebesgue integral.
If $g : X \to E$ is a simple function then its Bochner integral is~\href{https://github.com/leanprover-community/mathlib/blob/2cbaa9cc29bc44812b679810352ed383f35dcf75/src/measure_theory/bochner_integration.lean#L211}{\link}\
\[\int g\, \mathrm{d}\mu= \int g(x)\, \mathrm{d}\mu(x)=
\sum_{\substack{y \in g(X) \text{ s.t.}\\ \mu(g^{-1}\{y\}) < \infty}} \mu(g^{-1}\{y\})\cdot y \in E.\]
The symbol $\cdot$ denotes the scalar multiplication in the Banach space $E$.
Since this definition respects $\mu$-a.e.~equivalence,
we can also define the Bochner integral on the simple $L^1$ functions.~\href{https://github.com/leanprover-community/mathlib/blob/2cbaa9cc29bc44812b679810352ed383f35dcf75/src/measure_theory/bochner_integration.lean#L756}{\link}\
On the simple $L^1$ functions this definition is continuous~\href{https://github.com/leanprover-community/mathlib/blob/2cbaa9cc29bc44812b679810352ed383f35dcf75/src/measure_theory/bochner_integration.lean#L791}{\link}\
and the simple $L^1$ functions are dense in all $L^1$ functions,~\href{https://github.com/leanprover-community/mathlib/blob/2cbaa9cc29bc44812b679810352ed383f35dcf75/src/measure_theory/bochner_integration.lean#L710}{\link}\
and $E$ is complete, so we can extend the Bochner integral to all $L^1$ functions.~\href{https://github.com/leanprover-community/mathlib/blob/2cbaa9cc29bc44812b679810352ed383f35dcf75/src/measure_theory/bochner_integration.lean#L885}{\link}\

Finally, for an arbitrary function $f:X \to E$ we define its Bochner integral to be
$0$ if $f$ is not integrable and the integral of $[f]\in L^1(X,\mu;E)$ otherwise.~\href{https://github.com/leanprover-community/mathlib/blob/2cbaa9cc29bc44812b679810352ed383f35dcf75/src/measure_theory/bochner_integration.lean#L960}{\link}\
If we take the integral of a integrable function $X\to \R_{\ge0}$, then the two choices agree:~\href{https://github.com/leanprover-community/mathlib/blob/2cbaa9cc29bc44812b679810352ed383f35dcf75/src/measure_theory/bochner_integration.lean#L1250}{\link}\
\[\int^- f\, \mathrm{d}\mu = \int f\, \mathrm{d}\mu.\]
However, when $f$ is not integrable, then the equality does not hold:
the LHS is $\infty$, while the RHS is $0$.
We then prove useful properties about the Bochner integral,
such as the dominated convergence theorem:~\href{https://github.com/leanprover-community/mathlib/blob/2cbaa9cc29bc44812b679810352ed383f35dcf75/src/measure_theory/bochner_integration.lean#L1120}{\link}\
\begin{lstlisting}
theorem tendsto_integral_of_dominated_convergence
  {F : ℕ → X → E} {f : X → E} (bound : X → ℝ)
  (F_measurable : ∀ n, ae_measurable (F n) μ)
  (f_measurable : ae_measurable f μ)
  (bound_integrable : integrable bound μ)
  (h_bound : ∀ n, ∀ᵐ x ∂μ, ∥F n x∥ ≤ bound x)
  (h_lim : ∀ᵐ x ∂μ, tendsto (λ n, F n x) at_top (𝓝 (f x))) :
  tendsto (λ n, ∫ x, F n x ∂μ) at_top (𝓝 (∫ x, f x ∂μ))
\end{lstlisting}
In this statement, "∀ᵐ x ∂μ, P x" means that "P x" holds for "μ"-almost all "x", in other words that
"μ { x | ¬ P x } = 0",
and "tendsto (λ n, g n) at_top (𝓝 x)" means that "g n" tends to "x" as "n" tends to "∞".

Other results about integration include part 1~\href{https://github.com/leanprover-community/mathlib/blob/2cbaa9cc29bc44812b679810352ed383f35dcf75/src/measure_theory/set_integral.lean#L501}{\link}\ and 2~\href{https://github.com/leanprover-community/mathlib/blob/2cbaa9cc29bc44812b679810352ed383f35dcf75/src/measure_theory/interval_integral.lean#L1302}{\link}\ of the fundamental theorem of calculus,
but we will not use these results in this paper.

\section{Products of Measures}
\label{section:products}

To prove that the Haar measure is essentially unique, we need to work with product measures
and iterated integrals.
In this section we will define the product measure and prove Tonelli's and Fubini's theorem.
(Sometimes both theorems are called Fubini's theorem, but in this paper we will distinguish them.)
Tonelli's theorem characterizes the Lebesgue integral for nonnegative functions in the product space,
and Fubini's theorem does the same (under stronger conditions) for the Bochner integral.
By symmetry, both theorems also give sufficient conditions for swapping the order of integration
when working with iterated integrals.

The contents of this section have been formalized before in Isabelle.
Product measures, Tonelli's theorem and Fubini's theorem for the Lebesgue integral are discussed
in the paper~\cite{holzl2011measure}. Fubini's theorem for the Bochner integral is significantly
harder to prove, but has also been formalized in Isabelle after the appearance of
that paper.~\href{https://isabelle.in.tum.de/library/HOL/HOL-Analysis/Bochner_Integration.html}{\link}\
We could not find a good source for the proof of Fubini's theorem for the Bochner integral,
so we read through the proof in the Isabelle formalization for key ideas in various intermediate lemmas.
For this reason, we will provide detailed proofs of the results in the section.

In this section we will restrict our attention to $\upsigma$-finite measures,
since product measures are much nicer for $\upsigma$-finite measures,
and most of the results do not hold without this condition.
We say that a measure $\mu$ on a space $X$ is $\upsigma$-finite if $X$ can be written as a union
of countably many subsets $(A_i)_i$ with $\mu(A_i)<\infty$ for all $i$.~\href{https://github.com/leanprover-community/mathlib/blob/2cbaa9cc29bc44812b679810352ed383f35dcf75/src/measure_theory/measure_space.lean#L1570}{\link}\

Before we define products, we need some to know that limits of measurable functions are measurable.

\begin{lemma}\label{lem:measurable-limit}Let $X$ be a measurable space.
  \begin{enumerate}
    \item \label{part:inf_sup}
    Given a sequence of measurable functions $f_n:X\to [0,\infty]$ for $n\in\N$, the
    pointwise suprema and infima $f(x\vcentcolon=\sup\{f_n(x)\}_n$ and $g(x\vcentcolon=\inf\{f_n(x)\}_n$ are measurable.~\href{https://github.com/leanprover-community/mathlib/blob/2cbaa9cc29bc44812b679810352ed383f35dcf75/src/measure_theory/borel_space.lean#L809}{\link}\
    \item \label{part:lim_real}
    Given a sequence of measurable functions $f_n:X\to[0,\infty]$ for $n\in\N$
    that converge pointwise to a function $f$, we can conclude that $f$ is measurable.~\href{https://github.com/leanprover-community/mathlib/blob/2cbaa9cc29bc44812b679810352ed383f35dcf75/src/measure_theory/borel_space.lean#L1259}{\link}\
    \item \label{part:lim_metric} Let $Y$ be a metric space.
    Given a sequence of measurable functions $f_n:X\to Y$ for $n\in\N$
    that converge pointwise to a function $f$, we can conclude that $f$ is measurable.~\href{https://github.com/leanprover-community/mathlib/blob/2cbaa9cc29bc44812b679810352ed383f35dcf75/src/measure_theory/borel_space.lean#L1282}{\link}\
  \end{enumerate}
\end{lemma}
\begin{proof}
  For part~\ref{part:inf_sup},
  note that the collection of intervals $[0,x)$ generate the $\upsigma$-algebra on $[0,\infty]$,
  as do the intervals $(x,\infty]$.
  The conclusion follows from the fact that measurable sets are closed under countable unions
  and observing that
  $f^{-1}([0,x))=\bigcup_n f_n^{-1}([0,x))$ and
  $g^{-1}((x,\infty])=\bigcup_n f_n^{-1}((x,\infty])$

  Part~\ref{part:lim_real} follows by observing that $f(x)=\liminf_n f_n(x) = \sup_n\inf_{m\ge n} f_m(x)$,
  and applying both claims in part~\ref{part:inf_sup}.

  For part~\ref{part:lim_metric}, let $C$ be a closed set in $Y$.
  The function $d({-},C):Y\to[0,\infty]$ that assigns to each point the (minimal) distance to $C$ is
  continuous. Therefore, we know that $h(x) \vcentcolon= d(f(x),C)$ is the limit of the measurable functions
  $d(f_n(x),C)$ as $n\to\infty$. By part~\ref{part:lim_real} we know that $h$ is measurable.
  Therefore $f^{-1}(C)=h^{-1}(\{0\})$ is measurable (the equality holds because $C$ is closed).
\end{proof}

If $X$ and $Y$ are measurable spaces,
the $\upsigma$-algebra on $X\times Y$ is the smallest $\upsigma$-algebra that make the projection
$X\times Y\to X$ and $X\times Y\to Y$ measurable.~\href{https://github.com/leanprover-community/mathlib/blob/2cbaa9cc29bc44812b679810352ed383f35dcf75/src/measure_theory/measurable_space.lean#L661}{\link}\
Alternatively, it can be characterized as the $\upsigma$-algebra generated by sets of the form
$A\times B$ for measurable $A\subseteq X$ and $B\subseteq Y$.~\href{https://github.com/leanprover-community/mathlib/blob/2cbaa9cc29bc44812b679810352ed383f35dcf75/src/measure_theory/prod.lean#L135}{\link}\

For the rest of this section, $\mu$ is a $\upsigma$-finite measure on $X$
and $\nu$ is a $\upsigma$-finite measure on $Y$.
For a set $A\subseteq X\times Y$ we write $A_x=\{y \in Y \mid (x,y) \in A\}$ for a slice of $Y$.
In Lean, this is written as "prod.mk x ⁻¹' A".
We want to define the product measure evaluated at $A$
as an integral over the volume of the slices of $A$.
To make sure that this makes sense, we need to prove that this is a measurable function.~\href{https://github.com/leanprover-community/mathlib/blob/2cbaa9cc29bc44812b679810352ed383f35dcf75/src/measure_theory/prod.lean#L167}{\link}\
\begin{lstlisting}
lemma measurable_measure_prod_mk_left [sigma_finite ν] {A : set (X × Y)}
  (hA : is_measurable A) : measurable (λ x, ν (prod.mk x ⁻¹' A))
\end{lstlisting}
This lemma crucially depends on the fact that $\nu$ is $\upsigma$-finite,
but we omit the proof here. 
We can now define the product measure $\mu\times \nu$ on $X\times Y$ as~\href{https://github.com/leanprover-community/mathlib/blob/2cbaa9cc29bc44812b679810352ed383f35dcf75/src/measure_theory/prod.lean#L317}{\link}\
\begin{align}\label{eq:prod-def}
  (\mu\times\nu)(A) &= \int^- \nu(A_x)\, \mathrm{d}\mu(x)
\end{align}

It is not hard to see that $(\mu\times\nu)(A\times B)=\mu(A)\nu(B)$
for measurable $A\subseteq X$ and $B\subseteq Y$~\href{https://github.com/leanprover-community/mathlib/blob/2cbaa9cc29bc44812b679810352ed383f35dcf75/src/measure_theory/prod.lean#L330}{\link}\
and that $\mu\times\nu$ is itself $\upsigma$-finite.~\href{https://github.com/leanprover-community/mathlib/blob/2cbaa9cc29bc44812b679810352ed383f35dcf75/src/measure_theory/prod.lean#L421}{\link}\

Also, when we pushforward the measure $\mu\times\nu$ across the map $X\times Y\to Y\times X$
that swaps the coordinates, we get the measure $\nu\times\mu$~\href{https://github.com/leanprover-community/mathlib/blob/2cbaa9cc29bc44812b679810352ed383f35dcf75/src/measure_theory/prod.lean#L452}{\link}\
(this equality is easily checked because both measures are equal on rectangles).
This immediately gives the symmetric version of \eqref{eq:prod-def}:~\href{https://github.com/leanprover-community/mathlib/blob/2cbaa9cc29bc44812b679810352ed383f35dcf75/src/measure_theory/prod.lean#L459}{\link}\
\[(\mu\times\nu)(A) = \int^- \mu(A^y)\, \mathrm{d}\nu(y),\]
where $A^y=\{x \in X \mid (x,y) \in A\}$.

For a function $f:X\times Y \to Z$ define $f_x:Y\to Z$ by $f_x(y)\vcentcolon=f(x,y)$. If $Z=[0,\infty]$ let
$I^-_f:X\to [0,\infty]$ be defined by
$I^-_f(x)\vcentcolon=\int^-_Y f_x\, \mathrm{d}\nu=\int^-_Y f(x,y)\, \mathrm{d}\nu(y)$.

We can now formulate Tonelli's theorem.

\begin{theorem}[Tonelli's theorem]\label{thm:tonelli}
  Let $f:X\times Y\to[0,\infty]$ be a measurable function.
  Then~\href{https://github.com/leanprover-community/mathlib/blob/2cbaa9cc29bc44812b679810352ed383f35dcf75/src/measure_theory/prod.lean#L573}{\link}~\href{https://github.com/leanprover-community/mathlib/blob/2cbaa9cc29bc44812b679810352ed383f35dcf75/src/measure_theory/prod.lean#L610}{\link}\
  \[\int^-_{X\times Y} f\, \mathrm{d}(\mu\times\nu) =
  \int^-_X\!\!\! \int^-_Y f(x,y)\, \mathrm{d}\nu(y) \mathrm{d}\mu(x)
  = \int^-_Y\!\!\! \int^-_X f(x,y)\, \mathrm{d}\mu(x) \mathrm{d}\nu(y),\]
  and all the functions in the integrals above are measurable.~\href{https://github.com/leanprover-community/mathlib/blob/2cbaa9cc29bc44812b679810352ed383f35dcf75/src/measure_theory/prod.lean#L200}{\link}~\href{https://github.com/leanprover-community/mathlib/blob/2cbaa9cc29bc44812b679810352ed383f35dcf75/src/measure_theory/prod.lean#L226}{\link}\
\end{theorem}
\begin{proof}
  We check the first equality and the measurability of $I^-_f$.
  The other claims follow by symmetry. For both these statement we use
  \emph{induction on measurable functions into $[0,\infty]$},~\href{https://github.com/leanprover-community/mathlib/blob/2cbaa9cc29bc44812b679810352ed383f35dcf75/src/measure_theory/integration.lean#L1711}{\link}\
  which states that to prove a statement $P(f)$ for all measurable functions $f$ into $[0,\infty]$,
  it is sufficient to show
  \begin{enumerate}
    \item if $\chi_A$ is the characteristic function of a measurable set and $c\in[0,\infty]$, we have $P(c\chi_A)$;
    \item if $f_1$ and $f_2$ are measurable functions with $P(f_1)$ and $P(f_2)$, then $P(f_1+f_2)$;
    \item if $(f_i)_i$ is a monotone sequence of measurable functions such that $P(f_i)$ for all $i$,
    we have $P(\sup_i f_i)$.
  \end{enumerate}
  We first prove that $I^-_f$ is measurable by induction on $f$.

  In the induction base, if $f=c\chi_A$, then $\int^-_Y f(x,y)\, \mathrm{d}\nu(y)=c\nu(A_x)$,
  so the measurability follows from "measurable_measure_prod_mk_left".
  The two induction steps follow from the additivity of the integral
  and the monotone convergence theorem.

  We now prove the first equality in Theorem~\ref{thm:tonelli}, also by induction on $f$.
  In the induction base, if $f=c\chi_A$,
  then it is not hard to see that both sides reduce to $c(\mu\times\nu)(A)$.
  The first induction step follows by applying additivity thrice,
  and the second induction step follows by applying the monotone convergence theorem thrice.
\end{proof}

Fubini's theorem (sometimes called Fubini--Tonelli's theorem) states something similar for the Bochner integral.
However, it is a bit harder to state and much harder to prove.
For a function $f:X\times Y \to E$ define $I_f:X\to E$ by
$I_f(x)\vcentcolon=\int_Y f_x\, \mathrm{d}\nu=\int_Y f(x,y)\, \mathrm{d}\nu(y)$.
\begin{theorem}[Fubini's theorem for the Bochner integral]\label{thm:fubini}
  Let $E$ be a second-countable Banach space and $f:X\times Y \to E$ be a function.
  \begin{enumerate}
    \item\label{part:fubini-1}
    If $f$ is $\mu\times\nu$-a.e.~measurable then $f$ is integrable
    iff the following two conditions hold:~\href{https://github.com/leanprover-community/mathlib/blob/2cbaa9cc29bc44812b679810352ed383f35dcf75/src/measure_theory/prod.lean#L699}{\link}\
    \begin{itemize}
      \item for almost all $x\in X$ the function $f_x$ is $\nu$-integrable;
      \item The function $I_{\|f\|}$ is $\mu$-integrable.
    \end{itemize}
    \item\label{part:fubini-2} If $f$ is integrable, then~\href{https://github.com/leanprover-community/mathlib/blob/2cbaa9cc29bc44812b679810352ed383f35dcf75/src/measure_theory/prod.lean#L846}{\link}~\href{https://github.com/leanprover-community/mathlib/blob/2cbaa9cc29bc44812b679810352ed383f35dcf75/src/measure_theory/prod.lean#L870}{\link}\
    \[\int_{X\times Y} f\, \mathrm{d}(\mu\times\nu) =
    \int_X\! \int_Y f(x,y)\, \mathrm{d}\nu(y) \mathrm{d}\mu(x)
    = \int_Y\! \int_X f(x,y)\, \mathrm{d}\mu(x) \mathrm{d}\nu(y),\]
  \end{enumerate}
  Moreover, all the functions in the integrals above are a.e.~measurable.~\href{https://github.com/leanprover-community/mathlib/blob/2cbaa9cc29bc44812b679810352ed383f35dcf75/src/measure_theory/prod.lean#L290}{\link}~\href{https://github.com/leanprover-community/mathlib/blob/2cbaa9cc29bc44812b679810352ed383f35dcf75/src/measure_theory/prod.lean#L303}{\link}\
\end{theorem}
Note that Part~\ref{part:fubini-1} also has a symmetric counterpart,~\href{https://github.com/leanprover-community/mathlib/blob/2cbaa9cc29bc44812b679810352ed383f35dcf75/src/measure_theory/prod.lean#L707}{\link}\
which we omit here. Note that the middle term in Part~\ref{part:fubini-2} is just
$\int_X I_f\, \mathrm{d}\mu(x)$
\begin{proof}
  \textbf{Measurability}.
  First suppose that $f$ is measurable. In this case, we show that
  the function $I_f$ is measurable.
  We can approximate any measurable function $f$ into a $E$ by a
  sequence $(s_n)_n$ of simple functions~\href{https://github.com/leanprover-community/mathlib/blob/2cbaa9cc29bc44812b679810352ed383f35dcf75/src/measure_theory/simple_func_dense.lean#L106}{\link}\ such that for all $z$ we have
  $\|s_n(z)\|\le 2\|f(z)\|$~\href{https://github.com/leanprover-community/mathlib/blob/2cbaa9cc29bc44812b679810352ed383f35dcf75/src/measure_theory/simple_func_dense.lean#L162}{\link}\ and $s_n(z)\to f(z)$ as $n\to\infty$.~\href{https://github.com/leanprover-community/mathlib/blob/2cbaa9cc29bc44812b679810352ed383f35dcf75/src/measure_theory/simple_func_dense.lean#L132}{\link}\
  Now define $g_n:X\to E$ by
  \[g_n(x)=\begin{cases}I_{s_n}(x) &
    \text{if $f_x$ is $\nu$-integrable} \\
  0 & \text{otherwise.}\end{cases}\]
  Note that $g_n$ is similar to $I_{s_n}$, except that we set it to $0$ whenever
  $f_x$ is not integrable. This is required to ensure that $g_n$ converges to $I_f$ (see below).
  We first check that $g_n$ is measurable.
  Note that the set
  \[\{ x \mid \text{$f_x$ is $\nu$-integrable}\} =
  \{ x \mid \textstyle\int_Y \|f(x,y)\|\, \mathrm{d}\nu(y)<\infty\} \]
  is measurable, using Theorem~\ref{thm:tonelli}. Therefore, to show that $g_n$ is measurable,
  it suffices to show that $I_{s_n}$ is measurable,
  which is true, since it is a finite sum of functions of the form $x\mapsto c\nu(A_x)$
  (using that $s_n$ is simple).

  Secondly, we check that $g_n$ converges to $I_f$ pointwise. Let $x\in X$.
  If $f_x$ is not integrable, it is trivially true, so assume that $f_x$ is integrable.
  Then $(s_n)_x$ is integrable for all $n$, and by the dominated convergence theorem
  $I_{s_n}(x)$ converges to $I_f(x)$.

  We conclude that $I_f$ is measurable by Lemma~\ref{lem:measurable-limit}.

  Finally, if $f$ is $\mu\times\nu$-a.e.~measurable, then it is not hard to show that $I_f$ is
  $\mu$-a.e.~measurable.

  \textbf{Integrability}.
  We now prove Part~\ref{part:fubini-1}.
  By Tonelli's theorem, we know that $f$ is $\mu\times\nu$-integrable iff
  $I^-_{\|f\|}$ is $\mu$-integrable.
  Note that $I^-_{\|f\|}$ is almost the same as $I_{\|f\|}$.
  The difference is that when $f_x$ is not integrable, $I_{\|f\|}(x)=0$ but $I^-_{\|f\|}=\infty$.
  This means that if $I^-_{\|f\|}$ is $\mu$-integrable,
  then $f_x$ must almost always be $\nu$-integrable
  (otherwise the integrand is infinite on a set of positive measure) and
  $I_{\|f\|}$ must be $\mu$-integrable.
  For the other direction in Part~\ref{part:fubini-1},
  if $f_x$ is almost always integrable, then $I_{\|f\|}=I^-_{\|f\|}$ almost everywhere,
  so if the former is $\mu$-integrable, then so is the latter.

  \textbf{Equality}.
  To prove the first equality in Part~\ref{part:fubini-2}, we use the following
  \emph{induction principle for integrable functions}.~\href{https://github.com/leanprover-community/mathlib/blob/2cbaa9cc29bc44812b679810352ed383f35dcf75/src/measure_theory/set_integral.lean#L338}{\link}\
  It states that to prove a statement $P(f)$ for all $\tau$-integrable functions $f:Z\to E$,
  it is sufficient to show
  \begin{itemize}
    \item if $A$ is measurable with $\tau(A)<\infty$ and $c\in E$, then $P(c\chi_A)$;
    \item if $f_1$ and $f_2$ are $\tau$-integrable functions with $P(f_1)$ and $P(f_2)$, then $P(f_1+f_2)$;
    \item If $f_1$ is integrable with $P(f_1)$ and $f_1 = f_2$ $\tau$-a.e.~then $P(f_2)$.
    \item The set $\{ f \in L^1(Z,\tau;E) \mid P(f) \}$ is closed.
  \end{itemize}
  We now prove that
  \begin{align}\label{eq:fubini}
    \int_{X\times Y} f\, \mathrm{d}(\mu\times\nu) &=
  \int_X\! \int_Y f(x,y)\, \mathrm{d}\nu(y) \mathrm{d}\mu(x)
  \end{align}
  by induction on $f$.

  \begin{itemize}
    \item If $f=c\chi_A$ then it is easy to see that both sides of \eqref{eq:fubini}
    reduce to $c\cdot(\mu\times\nu)(A)$.
    \item It is clear that if \eqref{eq:fubini} holds for $f_1$ and $f_2$, then it also holds for $f_1+f_2$.
    \item Suppose that \eqref{eq:fubini} holds for $f_1$ and that $f_1=f_2$ almost everywhere.
    Since integrals are defined up to a.e.~equality, we know that
    \[\int_{X\times Y} f_1\, \mathrm{d}(\mu\times\nu)=\int_{X\times Y} f_2\, \mathrm{d}(\mu\times\nu).\]
    It therefore suffices to show that
    \[\int_X\! \int_Y f_1(x,y)\, \mathrm{d}\nu(y) \mathrm{d}\mu(x) = \int_X\! \int_Y f_2(x,y)\, \mathrm{d}\nu(y) \mathrm{d}\mu(x).\]
    For this it is sufficient to show that $I_{f_1}(x)=I_{f_2}(x)$ for almost all $x$.
    Since $f_1(z)=f_2(z)$ for almost all $z$, it is not hard to see~\href{https://github.com/leanprover-community/mathlib/blob/2cbaa9cc29bc44812b679810352ed383f35dcf75/src/measure_theory/prod.lean#L401}{\link}\ that
    for almost all $x$ we have that for almost all $y$ the equality $f_1(x,y)=f_2(x,y)$ holds.
    From this we easily get that $I_{f_1}(x)=I_{f_2}(x)$ for almost all $x$.
    \item We need to show that the set of $f \in L^1(Z,\tau;E)$ for which \eqref{eq:fubini} holds is closed.
    We know that taking the integral of an $L^1$ function is a continuous operation,~\href{https://github.com/leanprover-community/mathlib/blob/2cbaa9cc29bc44812b679810352ed383f35dcf75/src/measure_theory/bochner_integration.lean#L1062}{\link}\
    which means that both sides of \eqref{eq:fubini} are continuous in $f$,
    which means that the set where these functions are equal is closed.
  \end{itemize}
\end{proof}

\section{Existence of The Haar Measure}
\label{section:haar-def}

In this section, we define the left Haar measure,
and discuss some design decisions when formalizing these definitions.
We also discuss how to obtain a right Haar measure from this definition.

Throughout this section, we assume that "G" is a locally compact group,
equipped with the Borel $\upsigma$-algebra. We will write the group operation multiplicatively.
We will define the Haar measure and show that it is a left invariant regular measure on "G".
In the formalization,
all intermediate definitions and results are given with weaker conditions on "G",
following the design in \textsf{mathlib} that all results should be in the most general form (within reason).

In Lean, the precise conditions on "G" are written as follows.
\begin{lstlisting}
{G : Type*} [topological_space G] [t2_space G] [locally_compact_space G]
  [group G] [measurable_space G] [borel_space G] [topological_group G]
\end{lstlisting}

The concept of a left invariant measure is defined as follows in \textsf{mathlib}.~\href{https://github.com/leanprover-community/mathlib/blob/2cbaa9cc29bc44812b679810352ed383f35dcf75/src/measure_theory/group.lean#L38}{\link}\
\begin{lstlisting}
def is_left_invariant (μ : set G → ennreal) : Prop :=
∀ (g : G) {A : set G} (h : is_measurable A), μ ((λ h, g * h) ⁻¹' A) = μ A
\end{lstlisting}

Note that the preimage \lstinline{(λ h, g * h) ⁻¹' A} is equal to
the image \lstinline{(λ h, g⁻¹ * h) '' A}.
We use preimages for translations, since preimages are nicer to work with.
For example, the fact that \lstinline{(λ h, g * h) ⁻¹' A} is measurable
follows directly from the fact that multiplication (on the left) is a measurable function.

Next we give the definition of a regular measure~\href{https://github.com/leanprover-community/mathlib/blob/2cbaa9cc29bc44812b679810352ed383f35dcf75/src/measure_theory/borel_space.lean#L1387}{\link}\
\begin{lstlisting}
structure regular (μ : measure X) : Prop :=
(lt_top_of_is_compact : ∀ ⦃K : set X⦄, is_compact K → μ K < ⊤)
(outer_regular : ∀ ⦃A : set X⦄, is_measurable A →
  (⨅ (U : set X) (h : is_open U) (h2 : A ⊆ U), μ U) ≤ μ A)
(inner_regular : ∀ ⦃U : set X⦄, is_open U →
  μ U ≤ ⨆ (K : set X) (h : is_compact K) (h2 : K ⊆ U), μ K)
\end{lstlisting}
If "μ" is a regular measure,
then the inequalities in the last two fields are always equalities.
This means that a measure "μ" is regular if it is finite on compact sets,
its value on a measurable sets "A" is equal to
the infimum ("⨅") of its value on all open sets containing "A",
and finally its value on an open set "U" is the supremum ("⨆") of its values on
all compacts subsets of "U".

We are now ready to start the definition of the Haar measure.
Given a compact set "K" in "G",
we start by giving a rough approximation of the ``size'' of "K"
by comparing it to a reference set "V", which is an open neighborhood of "(1 : G)".
We can consider all left translates "(λ h, g * h) ⁻¹' V" of "V", which is an
open covering of "K". Since "K" is compact, we only need finitely many left
translates of "V" to cover "K".~\href{https://github.com/leanprover-community/mathlib/blob/2cbaa9cc29bc44812b679810352ed383f35dcf75/src/topology/algebra/group.lean#L563}{\link}\
Let "index K V" be the smallest number of left-translates of "V"
needed to cover "K".
This is often denoted $(K : V)$.
We did not use this in Lean code, since it conflicts with the typing-notation in Lean.
In Lean, this notion is defined for arbitrary sets "K" and "V",
and it is defined to be 0 if there is no finite number of left translates covering "K":~\href{https://github.com/leanprover-community/mathlib/blob/2cbaa9cc29bc44812b679810352ed383f35dcf75/src/measure_theory/haar_measure.lean#L73}{\link}\
\begin{lstlisting}
def index (K V : set G) : ℕ :=
Inf (finset.card '' {t : finset G | K ⊆ ⋃ g ∈ t, (λ h, g * h) ⁻¹' V })
\end{lstlisting}
For the rest of this section, we fix "K₀" as an arbitrary compact set with non-empty interior (in \textsf{mathlib} this is denoted "K₀ : positive_compacts G").
We then define a weighted version of the index:~\href{https://github.com/leanprover-community/mathlib/blob/2cbaa9cc29bc44812b679810352ed383f35dcf75/src/measure_theory/haar_measure.lean#L89}{\link}\
\begin{lstlisting}
def prehaar (K₀ U : set G) (K : compacts G) : ℝ :=
(index K U : ℝ) / index K₀ U
\end{lstlisting}
\begin{lemma}
\label{lem:prehaar-properties}
Denote "prehaar K₀ U K" by $h_U(K)$ (recall that "K₀" is fixed). Then we have
\begin{itemize}
\item $(K : U) \le (K : K_0)\cdot (K_0 : U)$;~\href{https://github.com/leanprover-community/mathlib/blob/2cbaa9cc29bc44812b679810352ed383f35dcf75/src/measure_theory/haar_measure.lean#L131}{\link}\
\item $0\le h_U(K) \le (K : K_0)$;~\href{https://github.com/leanprover-community/mathlib/blob/2cbaa9cc29bc44812b679810352ed383f35dcf75/src/measure_theory/haar_measure.lean#L230}{\link}\
\item $h_U(K_0)=1$;~\href{https://github.com/leanprover-community/mathlib/blob/2cbaa9cc29bc44812b679810352ed383f35dcf75/src/measure_theory/haar_measure.lean#L249}{\link}\
\item $h_U(xK)=h_U(K)$;~\href{https://github.com/leanprover-community/mathlib/blob/2cbaa9cc29bc44812b679810352ed383f35dcf75/src/measure_theory/haar_measure.lean#L265}{\link}\
\item if $K\subseteq K'$ then $h_U(K)\le h_U(K')$;~\href{https://github.com/leanprover-community/mathlib/blob/2cbaa9cc29bc44812b679810352ed383f35dcf75/src/measure_theory/haar_measure.lean#L242}{\link}\
\item if $h_U(K\cup K')\le h_U(K)+h_U(K')$~\href{https://github.com/leanprover-community/mathlib/blob/2cbaa9cc29bc44812b679810352ed383f35dcf75/src/measure_theory/haar_measure.lean#L253}{\link}\ and equality holds if
$KU^{-1}\cap K'U^{-1}=\emptyset$.~\href{https://github.com/leanprover-community/mathlib/blob/2cbaa9cc29bc44812b679810352ed383f35dcf75/src/measure_theory/haar_measure.lean#L260}{\link}\
\end{itemize}
\end{lemma}

To define the left Haar measure, we next want to define the ``limit''
of this quotient as "U" becomes a smaller and smaller open neighborhoods of 1.
This is not an actual limit, but we will emulate it by constructing a product space
that contains all these functions, and then use a compactness argument to find a ``limit''
in a large product space.

Consider the product
\[\prod_{K\subseteq G \text{ compact}}[0,(K : K_0)]\]
In Lean this is defined as a subspace of the topological space "compacts G → ℝ":~\href{https://github.com/leanprover-community/mathlib/blob/2cbaa9cc29bc44812b679810352ed383f35dcf75/src/measure_theory/haar_measure.lean#L100}{\link}\
\begin{lstlisting}
def haar_product (K₀ : set G) : set (compacts G → ℝ) :=
pi univ (λ K : compacts G, Icc 0 (index K K₀))
\end{lstlisting}
Here "pi univ" is the product of sets and "Icc" is an interval that is closed on both sides.

Note that by Tychonoff's theorem~\href{https://github.com/leanprover-community/mathlib/blob/2cbaa9cc29bc44812b679810352ed383f35dcf75/src/topology/subset_properties.lean#L603}{\link}\ "haar_product K₀" is compact,
and that every function "prehaar K₀ U : compacts G → ℝ" determines a point in "haar_product K₀".~\href{https://github.com/leanprover-community/mathlib/blob/2cbaa9cc29bc44812b679810352ed383f35dcf75/src/measure_theory/haar_measure.lean#L273}{\link}\

Given a open neighborhood "V" of 1, we can define the collection of points determined by
"prehaar K₀ U" for all "U ⊆ V" and take its closure:
\begin{lstlisting}
def cl_prehaar (K₀ : set G) (V : open_nhds_of (1 : G)) :
  set (compacts G → ℝ) :=
closure (prehaar K₀ '' { U : set G | U ⊆ V ∧ is_open U ∧ (1 : G) ∈ U })
\end{lstlisting}
Now we claim that the intersection of "cl_prehaar K₀ V" for all open neighborhoods "V" of 1
is non-empty.
\begin{lstlisting}
lemma nonempty_Inter_cl_prehaar (K₀ : positive_compacts G) : nonempty
  (haar_product K₀ ∩ ⋂ (V : open_nhds_of (1 : G)), cl_prehaar K₀ V)
\end{lstlisting}
\begin{proof}
Since "haar_product K₀" is compact, it is sufficient to show that if we have a finite collection of
open neighborhoods "t : finset (open_nhds_of 1)" the set
"haar_product K₀ ∩ ⋂ (V ∈ t), cl_prehaar K₀ V" is non-empty.~\href{https://github.com/leanprover-community/mathlib/blob/2cbaa9cc29bc44812b679810352ed383f35dcf75/src/topology/subset_properties.lean#L182}{\link}\ 
In this case we can explicitly give an element, namely "prehaar K₀ V₀", where "V₀ = ⋂ (V ∈ t), V".
\end{proof}

Finally, we can choose an arbitrary element in the set of the previous lemma,
which we call "chaar K₀", since it measures the size of the compact sets of "G":~\href{https://github.com/leanprover-community/mathlib/blob/2cbaa9cc29bc44812b679810352ed383f35dcf75/src/measure_theory/haar_measure.lean#L300}{\link}\
\begin{lstlisting}
def chaar (K₀ : positive_compacts G) (K : compacts G) : ℝ :=
classical.some (nonempty_Inter_cl_prehaar K₀) K
\end{lstlisting}

\begin{lemma}
  \label{lem:chaar-properties}
  Denote "chaar K₀ K" by $h(K)$. Then we have
  \begin{itemize}
  \item $0\le h(K)$~\href{https://github.com/leanprover-community/mathlib/blob/2cbaa9cc29bc44812b679810352ed383f35dcf75/src/measure_theory/haar_measure.lean#L311}{\link}\ and $h(\emptyset)=0$~\href{https://github.com/leanprover-community/mathlib/blob/2cbaa9cc29bc44812b679810352ed383f35dcf75/src/measure_theory/haar_measure.lean#L314}{\link}\ and $h(K_0)=1$;~\href{https://github.com/leanprover-community/mathlib/blob/2cbaa9cc29bc44812b679810352ed383f35dcf75/src/measure_theory/haar_measure.lean#L324}{\link}\
  \item $h(xK)=h(K)$;~\href{https://github.com/leanprover-community/mathlib/blob/2cbaa9cc29bc44812b679810352ed383f35dcf75/src/measure_theory/haar_measure.lean#L394}{\link}\
  \item if $K\subseteq K'$ then $h(K)\le h(K')$ (monotonicity);~\href{https://github.com/leanprover-community/mathlib/blob/2cbaa9cc29bc44812b679810352ed383f35dcf75/src/measure_theory/haar_measure.lean#L336}{\link}\
  \item if $h(K\cup K')\le h(K)+h(K')$ (subadditivity)~\href{https://github.com/leanprover-community/mathlib/blob/2cbaa9cc29bc44812b679810352ed383f35dcf75/src/measure_theory/haar_measure.lean#L349}{\link}\ and equality holds if
  $K\cap K'=\emptyset$ (additivity).~\href{https://github.com/leanprover-community/mathlib/blob/2cbaa9cc29bc44812b679810352ed383f35dcf75/src/measure_theory/haar_measure.lean#L364}{\link}\
  \end{itemize}
\end{lemma}
The idea in the proof is that the projections $f \mapsto f(K)$ are continuous functions.
Since these statements in this lemma hold for all $h_U$ by Lemma~\ref{lem:prehaar-properties},
they also hold for each point in each set "cl_prehaar K₀ V", and therefore for $h$.

This function "chaar K₀" measures the size of the compact sets of "G".
It has the properties of a \emph{content} on the compact sets, which is an additive, subadditive and
monotone function into the nonnegative reals~\cite[\S 53]{halmos1950measure}.
From this content we can define the Haar measure in three steps.

First, given an arbitrary content,
we can obtain its associated \emph{inner content} on the open sets,~\href{https://github.com/leanprover-community/mathlib/blob/2cbaa9cc29bc44812b679810352ed383f35dcf75/src/measure_theory/content.lean#L54}{\link}\
\begin{lstlisting}
def inner_content (h : compacts G → ennreal) (U : opens G) : ennreal :=
⨆ (K : compacts G) (h : K ⊆ U), h K
\end{lstlisting}
This inner content is monotone~\href{https://github.com/leanprover-community/mathlib/blob/2cbaa9cc29bc44812b679810352ed383f35dcf75/src/measure_theory/content.lean#L82}{\link}\ and countably subadditive.~\href{https://github.com/leanprover-community/mathlib/blob/2cbaa9cc29bc44812b679810352ed383f35dcf75/src/measure_theory/content.lean#L127}{\link}\
Also note that for every set "K" that is both compact and open, we have "inner_content h K = h K".~\href{https://github.com/leanprover-community/mathlib/blob/2cbaa9cc29bc44812b679810352ed383f35dcf75/src/measure_theory/content.lean#L67}{\link}\

Second, given an arbitrary inner content "ν", we can obtain its associated outer measure "μ"~\href{https://github.com/leanprover-community/mathlib/blob/2cbaa9cc29bc44812b679810352ed383f35dcf75/src/measure_theory/content.lean#L179}{\link}\
which is given by "μ A = ⨅ (U : set G) (hU : is_open U) (h : A ⊆ U), ν A".~\href{https://github.com/leanprover-community/mathlib/blob/2cbaa9cc29bc44812b679810352ed383f35dcf75/src/measure_theory/content.lean#L202}{\link}\

For open sets "U" we have that "μ U = ν U".~\href{https://github.com/leanprover-community/mathlib/blob/2cbaa9cc29bc44812b679810352ed383f35dcf75/src/measure_theory/content.lean#L186}{\link}\
Suppose that the inner content "ν" comes from the content "h".
We have observed that "μ" agrees with "ν" on their common domain (the open sets),
and that "ν" agrees with "h" on their common domain (the compact open sets).
One might be tempted to conclude that "μ" agrees with "h" on their common domain (the compact sets).
This is implicitly claimed in Gleason~\cite{gleason2010existence},
since the same name is used for these three functions.
However, this is not necessarily the case~\cite[p. 233]{halmos1950measure}.
The best we can say in general is that for compact "K"
we have "μ (interior K) ≤ h K ≤ μ K"~\href{https://github.com/leanprover-community/mathlib/blob/2cbaa9cc29bc44812b679810352ed383f35dcf75/src/measure_theory/content.lean#L190}{\link}~\href{https://github.com/leanprover-community/mathlib/blob/2cbaa9cc29bc44812b679810352ed383f35dcf75/src/measure_theory/content.lean#L194}{\link}\

If we apply these two steps to the function "h = chaar K₀"
we obtain the Haar outer measure "haar_outer_measure K₀" on "G".~\href{https://github.com/leanprover-community/mathlib/blob/2cbaa9cc29bc44812b679810352ed383f35dcf75/src/measure_theory/haar_measure.lean#L444}{\link}\
Next, we show that all Borel measurable sets on "G" are Carathéodory measurable
w.r.t. "haar_outer_measure K₀".~\href{https://github.com/leanprover-community/mathlib/blob/2cbaa9cc29bc44812b679810352ed383f35dcf75/src/measure_theory/haar_measure.lean#L509}{\link}\
This gives the \emph{left Haar measure}.~\href{https://github.com/leanprover-community/mathlib/blob/2cbaa9cc29bc44812b679810352ed383f35dcf75/src/measure_theory/haar_measure.lean#L537}{\link}\
We also scale the Haar measure so that it assigns measure 1 to "K₀".
\begin{lstlisting}
def haar_measure (K₀ : positive_compacts G) : measure G :=
(haar_outer_measure K₀ K₀.1)⁻¹ •
  (haar_outer_measure K₀).to_measure (haar_caratheodory_measurable K₀)
\end{lstlisting}

\begin{theorem}\label{thm:haar-measure}
  Let $\mu$ be the left Haar measure on $G$, defined above.
  Then $\mu$ is a left invariant~\href{https://github.com/leanprover-community/mathlib/blob/2cbaa9cc29bc44812b679810352ed383f35dcf75/src/measure_theory/haar_measure.lean#L546}{\link}\ and regular~\href{https://github.com/leanprover-community/mathlib/blob/2cbaa9cc29bc44812b679810352ed383f35dcf75/src/measure_theory/haar_measure.lean#L571}{\link}\ measure satisfying $\mu(K_0)=1$.~\href{https://github.com/leanprover-community/mathlib/blob/2cbaa9cc29bc44812b679810352ed383f35dcf75/src/measure_theory/haar_measure.lean#L555}{\link}\
\end{theorem}

We will show in Section~\ref{section:uniqueness} that the Haar measure is unique,
up to multiplication determined by a constant.
We finish this section by proving some properties about left invariant regular measures.

\begin{theorem}\label{thm:regular-measure}
  Let $\mu$ be a nonzero left invariant regular measure on $G$.
\begin{itemize}
  \item If $U$ is a nonempty open set, the $\mu(U)>0$~\href{https://github.com/leanprover-community/mathlib/blob/2cbaa9cc29bc44812b679810352ed383f35dcf75/src/measure_theory/group.lean#L181}{\link}\
  \item If $f : G \to \R$ is a nonnegative continuous function, not identically equal to 0. Then
  $\int^- f\, \mathrm{d}\mu>0$.~\href{https://github.com/leanprover-community/mathlib/blob/2cbaa9cc29bc44812b679810352ed383f35dcf75/src/measure_theory/group.lean#L189}{\link}\
\end{itemize}
\end{theorem}

In this section we have only discussed the left Haar measure, which is left invariant.
One can similarly define the notion of right invariance for a measure.
We can easily translate between these two notions.
For a measure $\mu$ define the measure $\check\mu$ by $\check\mu(A)=\mu(A^{-1})$.~\href{https://github.com/leanprover-community/mathlib/blob/2cbaa9cc29bc44812b679810352ed383f35dcf75/src/measure_theory/group.lean#L76}{\link}\

\begin{theorem}\label{thm:inv-measure}\mbox{}
\begin{itemize}
  \item The operation $\mu\mapsto\check\mu$ is involutive.~\href{https://github.com/leanprover-community/mathlib/blob/2cbaa9cc29bc44812b679810352ed383f35dcf75/src/measure_theory/group.lean#L86}{\link}\
  \item $\mu$ is left invariant iff $\check\mu$ is right invariant.~\href{https://github.com/leanprover-community/mathlib/blob/2cbaa9cc29bc44812b679810352ed383f35dcf75/src/measure_theory/group.lean#L133}{\link}\
  \item $\mu$ is regular iff $\check\mu$ is regular.~\href{https://github.com/leanprover-community/mathlib/blob/2cbaa9cc29bc44812b679810352ed383f35dcf75/src/measure_theory/group.lean#L104}{\link}\
\end{itemize}
\end{theorem}

Theorem~\ref{thm:inv-measure} shows that if $\mu$ is the left Haar measure on $G$,
then $\check\mu$ is a right Haar measure on $G$ (i.e.~a nonzero regular right invariant measure).

\section{Uniqueness of the Haar measure}
\label{section:uniqueness}

In this section we will show that the Haar measure is unique, up to multiplication by a constant.
For this proof, we followed the proof in Halmos~\cite[\S 59-60]{halmos1950measure}.
We will compare our proof with the proof given in Halmos at the end of the section.

Let $G$ be a second-countable locally compact group.
The main ingredient in the proof is that if we have two left invariant measures $\mu$ and $\nu$ on $G$,
we can transform expressions involving $\mu$ into expressions involving $\nu$
by applying a transformation to $G\times G$ that preserves the measure $\mu\times\nu$.
We need Tonelli's theorem to compute the integrals in this product.
Because $G$ is second countable,
the product $\upsigma$-algebra on $G\times G$ coincides with the Borel $\upsigma$-algebra.~\href{https://github.com/leanprover-community/mathlib/blob/2cbaa9cc29bc44812b679810352ed383f35dcf75/src/measure_theory/borel_space.lean#L512}{\link}\

Suppose $\mu$ and $\nu$ are left invariant and regular measures on $G$.
Since $G$ is second-countable, this also means that $\mu$ and $\nu$ are $\upsigma$-finite.~\href{https://github.com/leanprover-community/mathlib/blob/2cbaa9cc29bc44812b679810352ed383f35dcf75/src/measure_theory/borel_space.lean#L1447}{\link}\

We first show that $(x,y)\mapsto (x,xy)$ and $(x,y)\mapsto (yx,x^{-1})$ are
measure-preserving transformations.

\begin{lemma}\label{lem:measure-preserving}
  Let $\mu$ and $\nu$ be left invariant regular measures on $G$, then the transformations
  $S, T : G\times G \to G\times G$ given by $S(x,y)\vcentcolon=(x,xy)$ and $T(x,y)\vcentcolon=(yx,x^{-1})$
  preserve the measure $\mu\times \nu$, i.e. $S_*(\mu\times\nu)=\mu\times\nu=T_*(\mu\times\nu)$.~\href{https://github.com/leanprover-community/mathlib/blob/2cbaa9cc29bc44812b679810352ed383f35dcf75/src/measure_theory/prod_group.lean#L50}{\link}~\href{https://github.com/leanprover-community/mathlib/blob/2cbaa9cc29bc44812b679810352ed383f35dcf75/src/measure_theory/prod_group.lean#L89}{\link}\
\end{lemma}
\begin{proof}
  It suffices to show that the measures are equal on measurable rectangles,
  since these generate the $\upsigma$-algebra on $G\times G$.
  Let $A\times B\subseteq G\times G$ be a measurable rectangle. For $S$ we compute:
  \begin{align*}
    S_*(\mu\times\nu)(A\times B)&=(\mu\times\nu)(S^{-1}(A\times B))\\
    &=\int^- \nu(S(x,{-})^{-1}(A\times B))\, \mathrm{d}\mu(x)\\
    &=\int^- \chi_A(x)\nu(x^{-1}B)\, \mathrm{d}\mu(x)\\
    &=\int^- \chi_A(x)\nu(B)\, \mathrm{d}\mu(x) = \mu(A)\nu(B)=(\mu\times\nu)(A\times B).
  \end{align*}
  Note that $S$ is an equivalence with inverse $S^{-1}(x,y)=(x,x^{-1}y)$.
  If we define $R(x,y)\vcentcolon=(y,x)$, we saw in Section~\ref{section:products} that $R_*(\mu\times\nu)=\nu\times\mu$.
  The claim that $T$ preserves measure now follows from the observation that $T=S^{-1}RSR$,
  which is easy to check.
\end{proof}

For measurable sets $A$ its left translates have the same measure: $\mu(xA)=\mu(A)$.
This might not be true right translates $\mu(Ax)$, but we can say the following~\cite[\S 59, Th. D]{halmos1950measure}.

\begin{lemma}\label{lem:right-translations}
  Let $\mu$ be a left invariant regular measure on $G$ and
  suppose that $\mu(A)>0$ for a measurable set $A\subseteq G$.
  Then $\mu(A^{-1})>0$~\href{https://github.com/leanprover-community/mathlib/blob/2cbaa9cc29bc44812b679810352ed383f35dcf75/src/measure_theory/prod_group.lean#L115}{\link}\ and $\mu(Ax)>0$~\href{https://github.com/leanprover-community/mathlib/blob/2cbaa9cc29bc44812b679810352ed383f35dcf75/src/measure_theory/prod_group.lean#L146}{\link}\ for $x\in G$.
  Furthermore, the map $f(x)\vcentcolon=\mu(Ax)$ is measurable.~\href{https://github.com/leanprover-community/mathlib/blob/2cbaa9cc29bc44812b679810352ed383f35dcf75/src/measure_theory/prod_group.lean#L125}{\link}\
\end{lemma}
\begin{proof}
  Let $S$ and $T$ be as in Lemma~\ref{lem:measure-preserving}.
  For the first claim, we will show that $\mu(A^{-1})=0$ implies that $\mu(A)=0$.
  It suffices to show that $(\mu\times\mu)(A\times A)=0$, which we can compute using Lemma~\ref{lem:measure-preserving}:
  \begin{align*}
  (\mu\times\mu)(A\times A)
  &=(\mu\times\mu)(T^{-1}(A\times A))\\
  &=\int^{-}\mu((T^{-1}(A\times A))^y)\, \mathrm{d}\mu(y)\\
  &=\int^{-}\mu(y^{-1}A\cap A^{-1})\, \mathrm{d}\mu(y)
  \le\int^{-}\mu(A^{-1})\, \mathrm{d}\mu(y)
  = 0.
  \end{align*}
  The second claim now easily follows,
  since $\mu(A)>0$ implies that $\mu(Ax)=\mu((xA^{-1})^{-1})>0$.

  For the final claim, it is not hard to see that $\mu(Ax)=\mu((S^{-1}(G\times A)^x)$,
  which is measurable by the symmetric version of "measurable_measure_prod_mk_left"
  from Section~\ref{section:products}.
\end{proof}

The following is a technical lemma that allows us to rewrite $\nu$-integrals
into $\mu$-integrals. This is the key lemma to show the uniqueness of the Haar measure.~\cite[\S 60, Th. A]{halmos1950measure}.
\begin{lemma}\label{lem:left-invariant-lemma}
  Let $\mu$ and $\nu$ be a left invariant regular measures on $G$,
  let $K\subseteq G$ be compact with $\nu(K)>0$, and suppose that $f:G\to[0,\infty]$ is measurable.
  Then~\href{https://github.com/leanprover-community/mathlib/blob/2cbaa9cc29bc44812b679810352ed383f35dcf75/src/measure_theory/prod_group.lean#L167}{\link}\
  \[\mu(K)\cdot\int^-\frac{f(y^{-1})}{\nu(Ky)}\, \mathrm{d}\nu(y)=\int^- f\, \mathrm{d}\mu.\]
\end{lemma}
\begin{proof}
  First note that $0<\nu(Ky)<\infty$ for any $y$.
  The first inequality follows from Lemma~\ref{lem:right-translations}
  and the second inequality from the fact that $\nu$ is regular and $Kx^{-1}$ is compact.
  Let $g(y)\vcentcolon=\frac{f(y^{-1})}{\nu(Ky)}$, then by Lemma~\ref{lem:right-translations},
  $g$ is measurable. Now we compute
  \begin{align*}
    \mu(K)\cdot\int^-g(y)\, \mathrm{d}\nu(y)
    &=\int^-\chi_K(x)\, \mathrm{d}\mu(x)\cdot\int^-g(y)\, \mathrm{d}\nu(y)\\
    &=\int^-\!\!\!\int^-\chi_K(x)g(y)\, \mathrm{d}\nu(y)\, \mathrm{d}\mu(x)\\
    &=\int^-\!\!\!\int^-\chi_K(yx)g(x^{-1})\, \mathrm{d}\nu(y)\, \mathrm{d}\mu(x)&&\text{(by Lemma~\ref{lem:measure-preserving})}\\
    &=\int^-\nu(Kx^{-1})g(x^{-1})\, \mathrm{d}\mu(x)\\
    &=\int^-f(x)\, \mathrm{d}\mu(x).&&\qedhere
  \end{align*}
\end{proof}

As a consequence, we get the uniqueness of Haar measure: every left invariant measure is a multiple of
Haar measure.
\begin{theorem}\label{thm:haar-unique}
  Let $\nu$ be a left invariant regular measure and $K_0$ be a compact set with non-empty interior.
  If $\mu$ is the Haar measure on $G$ with $\mu(K_0)=1$, then $\nu = \nu(K_0)\cdot \mu$.~\href{https://github.com/leanprover-community/mathlib/blob/2cbaa9cc29bc44812b679810352ed383f35dcf75/src/measure_theory/haar_measure.lean#L597}{\link}\
\end{theorem}
In the code, this theorem is stated as follows (in the formalization we assume $\mu$ is $\upsigma$-finite instead of regular):
\begin{lstlisting}
theorem haar_measure_unique [sigma_finite μ] (hμ : is_left_invariant μ)
  (K₀ : positive_compacts G) : μ = μ K₀ • haar_measure K₀
\end{lstlisting}
\begin{proof}
  Let $A$ be any measurable set.
  The result follows from the following computation,
  where we apply Lemma~\ref{lem:left-invariant-lemma} twice,
  once with measures $(\nu,\mu)$ and once with measures $(\mu,\mu)$:
  \begin{align*}
    \nu(A)&=\int^- \chi_A\, \mathrm{d}\nu
    =\nu(K_0)\cdot\int^-\frac{\chi_A(y^{-1})}{\mu(K_0y)}\, \mathrm{d}\mu(y)\\
    &=\nu(K_0)\cdot \mu(K_0)\cdot \int^-\frac{\chi_A(y^{-1})}{\mu(K_0y)}\, \mathrm{d}\mu(y)
    =\nu(K_0)\cdot \int^- \chi_A\, \mathrm{d}\mu
    =\nu(K_0)\cdot \mu(A).\qedhere
  \end{align*}
\end{proof}

For this proof we followed Halmos~\cite[\S 59-60]{halmos1950measure},
but there are some differences.
Halmos gives a more general version of Lemma~\ref{lem:left-invariant-lemma},
for an arbitrary measurable group,
while we did it for a
second-countable locally compact group equipped with the Borel $\upsigma$-algebra (which forms a measurable group).
There are other proofs, like in Cohn~\cite[Theorem 9.2.6]{cohn2013measure},
that do not require $G$ to be second countable.
Another difference between our and Halmos' proof is that we assume that $K$ is compact,
while Halmos assumes that $K$ is measurable with $\nu(K)<\infty$.
However, in the proof he implicitly uses that also $\nu(Ky)<\infty$,
a fact that is not justified in the text,
and that I was unable to prove from the assumption that $\nu(K)<\infty$.
At this point I do not know whether this claim is true or not.
Lemma~\ref{lem:right-translations} claims something similar,
namely that if $\nu(K)>0$ then $\nu(Ky)>0$,
but the proof of this claim does not work to show finiteness of $\nu(Ky)$.
We patched this gap by assuming that $K$ is compact and $\nu$ is regular,
which implies that $\nu(Ky)<\infty$ because $Ky$ is also compact.

\section{Concluding Thoughts} 
\label{section:conclusion}

The formalization of product measures is about 900 lines of code,
and the formalization of Haar measure is about 1300 lines of code
(including blank lines and comments).
The proofs of these concepts use a wide variety of techniques,
from the reasoning in Banach spaces for Fubini's theorem,
to topological compactness arguments in the existence proof of the Haar measure,
to the calculcational proofs in the uniqueness of the Haar measure.
This formalization shows that these different proof techniques can be used and combined effectively
in \textsf{mathlib} to formalize graduate level mathematics.

This work can be extended in various ways.
One interesting project would be to prove the uniqueness for groups
that need not be second countable. This should not be too difficult,
but requires a version of Fubini's theorem for compactly supported functions on locally compact Hausdorff spaces.
Another project would be too explore the \emph{modular function},
which describes how the left and the right Haar measures differ on a locally compact group.
Furthermore, the Haar measure is used as a tool in many areas of mathematics,
and it would be interesting to investigate which of these topics are now amenible to formalization,
with the Haar measure as new tool in our toolbelt.

One area that should be feasible to formalize is abstract harmonic analysis.
Abstract harmonic analysis is the area of analysis on topological groups,
usually locally compact groups.
One of the core ideas is to generalize the Fourier transform
to an arbitrary locally compact abelian group.
The special case of the continuous Fourier transform has been formalized in HOL4~\cite{guan2020fourier},
but the abstract Fourier transform has never been formalized.
One specific goal could be to formallize \emph{Pontryagin duality}.
For a locally compact abelian group $G$ we can define
the \emph{Pontryagin dual} $\widehat{G}\vcentcolon=\Hom(G,T)$
consisting of the continuous group homomorphisms from $G$ to the circle group $T$.
Then Pontryagin duality states that the canonical map $G\to\widehat{\widehat{G}}$
(given by the evaluation function) is an isomorphism.
More ambitious targets include representation theory of compact groups,
Peter--Weyl theorem and Weyl's integral formula for compact Lie groups.

In \textsf{mathlib} we have not just defined binary product measures,
but also finitary product measures,~\href{https://github.com/leanprover-community/mathlib/blob/2cbaa9cc29bc44812b679810352ed383f35dcf75/src/measure_theory/pi.lean#L212}{\link}\
which can be used to define integrals over $\R^n$.
In \textsf{mathlib} this is done by defining a measure on the product space "Π i : ι, X i".
\begin{lstlisting}
def measure.pi {ι : Type*} [fintype ι] {X : ι → Type*}
  [∀ i, measurable_space (X i)] (μ : ∀ i, measure (X i)) :
  measure (Π i, X i)
\end{lstlisting}
The definition is conceptually very simple,
since you just iterate the construction of binary product measures,
but in practice it is quite tricky to do it in a way that is convenient to use.
For example, in a complicated proof, you do not want to worry about the fact that the spaces
$X\times(Y\times Z)$ and $(X\times Y)\times Z$ are not exactly the same, just equivalent.
One open question is whether we can formulate Fubini's theorem for finitary product measures
in a way that is convenient to use. It would be inconvenient to rewrite a finitary product of
measures into a binary product in the desired way,
and then apply Fubini's theorem for binary product measures.
We have an idea that we expect to be more promising, which is to define a new kind integral
that only some variables in the input of a function.
It takes a function "f : (Π i, π i) → E" and integrates away the variables in a subset "I".
The result is another function "(Π i, π i) → E" that is constant on all variables in "I":
\begin{lstlisting}
def integrate_away (μ : ∀ i, measure (π i)) (I : finset ι)
  (f : (Π i, π i) → E) : (Π i, π i) → E
\end{lstlisting}
This has not been incorporated in \textsf{mathlib},
but it is a promising experiment on a separate branch.~\href{https://github.com/leanprover-community/mathlib/blob/d0bab9ddb55bc885b6fd39d1eb21a053bd6e1d0c/src/measure_theory/pi.lean#L397}{\link}\
If we denote "integrate_away μ I f" as "∫⋯∫_I, f ∂μ",
then the following is a convenient way to state Fubini's theorem for finitary product measures.
\begin{lstlisting}
lemma integrate_away_union (f : (Π i, π i) → E) (I J : finset ι) :
  disjoint I J → ∫⋯∫_I ∪ J, f ∂μ = ∫⋯∫_ I, ∫⋯∫_J, f ∂μ ∂μ
\end{lstlisting}

\bibliography{references}

\end{document}